\RequirePackage{fix-cm}
\documentclass[a4paper,12pt]{svjour3}                     
\usepackage[margin=1in]{geometry}

\usepackage{multirow}

\usepackage{amsmath}
\usepackage{amssymb}
\usepackage{bm}
\usepackage{mathrsfs}
\usepackage[dvips]{graphicx}

\usepackage{algorithmic}
\usepackage{algorithm}

\usepackage{color}

\usepackage{url}

\usepackage{supertabular}

\newtheorem{thm}{Theorem}

\DeclareMathAlphabet{\mathsfsl}{OT1}{cmss}{m}{sl}

\renewcommand{\phi}{\varphi}

\newcommand{\Expect}{\operatorname{\mathbb{E}}}

\newcommand{\bx}{\mathbf{x}}

\newcommand{\ba}{\mathbf{a}}
\newcommand{\bb}{\mathbf{b}}
\newcommand{\by}{\mathbf{y}}

\newcommand{\bz}{\mathbf{z}}

\newcommand{\sA}{\mathcal{A}}

\newcommand{\bbC}{\mathcal{C}}

\def\reals{\mathbb{R}}

\def\calW{\mathcal{W}}

\def\calI{\mathcal{I}}

\def\bx{\mathbf{x}}

\def\bu{\mathbf{u}}
\def\bS{\mathbf{S}}
\def\b0{\mathbf{0}}

\def\bv{\mathbf{v}}
\def\bA{\mathbf{A}}

\def\calS{\mathcal{S}}
\def\bI{\mathbf{I}}
\def\sI{\mathcal{I}}

\def\Sp{\mathrm{Sp}}

\newcommand{\di}{{\,\mathrm{d}}}

\smartqed  
\usepackage{graphicx}
\begin{document}

\title{Phase retrieval of complex-valued objects via a randomized Kaczmarz method}



\author{Teng Zhang and Feng Yu
}


\institute{T. Zhang \at
              University of Central Florida, Department of Mathematics \\
              \email{teng.zhang@ucf.edu}             \\
              F. Yu \at
              University of Central Florida, Department of Mathematics \\
              \email{yfeng@Knights.ucf.edu}          
}

\date{Received: date / Accepted: date}

\maketitle

\begin{abstract}
  This paper investigates the convergence of the randomized Kaczmarz algorithm for the problem of phase retrieval of complex-valued objects. While this algorithm  has been studied for the real-valued case in \cite{10.1093/imaiai/iay005}, its generalization to the complex-valued case is nontrivial and has been left as a conjecture. This paper applies a different approach by establishing the connection between the convergence of the algorithm and the convexity of an objective function. Based on the connection, it demonstrates that when the sensing vectors are sampled uniformly from a unit sphere in $\bbC^n$ and the number of sensing vectors $m$ satisfies $m>O(n\log n)$ as $n, m\rightarrow\infty$, then this algorithm with a good initialization achieves linear convergence to the solution with high probability. The method can be applied to other statistical models of sensing vectors as well. A similar convergence result is established for the unitary model, where the sensing vectors are from the columns of random orthogonal matrices.
\end{abstract}
\section{Introduction}
%
%
%
%
This article concerns the phase retrieval problem as follows: let $\bz\in\bbC^n$ be an unknown vector,  given $m$ known sensing vectors $\{\ba_i\}_{i=1}^m\in\bbC^n$ and the observations
\begin{equation}\label{eq:problem}
y_i=|\ba_i^*\bz|, i=1,2,\cdots,m,
\end{equation}
then can we reconstruct $\bz$ from the observations $\{y_i\}_{i=1}^m$? 

Many algorithm has been proposed for this problem, including approaches based on convex relaxation to semidefinite optimization~\cite{Chai2011,Candes_PhaseLift,Waldspurger2015,Candes2014,Gross2015}, convex relaxation to linear program~\cite{pmlr-v54-bahmani17a,Goldstein2016,Hand2016,Hand20162,NIPS2018_8082}, nonconvex approaches based on Wirtinger flows, i.e., gradient flow in the complex setting~\cite{Candes7029630,NIPS2015_5743,Zhang:2016:PNP:3045390.3045499,NIPS2016_6319,cai2016,NIPS2016_6061,Soltanolkotabi2017,Chen2018,Candes7029630,Soltanolkotabi2017,7541725},  alternate minimization (Gerchberg-Saxton) algorithm and its variants~\cite{Gerchberg72,Fienup78,Fienup82,Bauschke03,Waldspurger2016,Netrapalli7130654,Zhang2017,zhang2020},  and algorithms based on Douglas-Rachford splitting~\cite{doi:10.1137/18M1170364}. This work investigates the randomized Kaczmarz, which is simple to implement and has shown competitive performance in simulations. This algorithm is first proposed by Wei in \cite{Wei_2015} and it is shown that the method performs comparably with the state-of-the-art Wirtinger flow methods,  when the
sensing vectors are from real or complex Gaussian distributions, or when they follow the unitary model or the coded diffraction pattern (CDP) model. The work also includes a preliminary convergence analysis. The convergence analysis is improved by Tan and Vershynin in \cite{10.1093/imaiai/iay005}, which shows that the algorithm is successful when  there are as many Gaussian measurements as the dimension, up to a constant factor, and the initialization is in a ``basin of linear convergence''. However, their results only apply when the signal $\bz$ and the measurement vectors $\{\ba_i\}_{i=1}^m$ are real-valued. As discussed in~\cite[Section 7.2]{10.1093/imaiai/iay005}, there is no  straightforward generalization of their technique from the real-valued case to the complex-valued case. In a related work~\cite{tan2019online}, Tan and Vershynin show that constant step
size online stochastic gradient descent (SGD) converges from arbitrary initializations for a
non-smooth, non-convex amplitude squared loss objective, and this online SGD is strongly reminiscent to the randomized Kaczmarz algorithm from numerical analysis. However, the analysis in \cite{tan2019online} is still based on the real-valued setting.

\subsection{Randomized Kaczmarz algorithm for solving linear systems}
The Kaczmarz method \cite{kaczmarz1379} is an iterative algorithm for solving a system of linear equations
\begin{equation}\label{eq:linearsystem}
\ba_i^*\bx=b_i, i=1,\cdots,m.
\end{equation}
In the $k$-th iteration, a linear equation (out of $m$ equations) is selected and the new estimate $\bx^{(k+1)}$ is obtained by projecting the current estimate $\bx^{(k)}$ to the hyperplane corresponding to the solution set of the linear equation. The deterministic version of the Kaczmarz method usually selects the linear equation in a cyclic manner, and the randomized Kaczmarz method selects the linear equation randomly. When the randomized
Kaczmarz method randomly picks up a system with the probability proportional to $1/\|\ba_i\|^2$, the randomized Kaczmarz method has been shown to converge linearly in \cite{Strohmer2008} with a rate of $1-\kappa(\bA)$, where $\kappa(\bA)$ is the condition number of $\bA=[\ba_1,\cdots,\ba_m]\in\reals^{m\times n}$.  For additional analysis on this method and its variants, we refer the readers to~\cite{NEEDELL2014199,Needell2016}.

\subsection{Randomized Kaczmarz algorithm for phase retrieval}
The randomized Kaczmarz algorithm can be generalized to the  phase retrieval problem \eqref{eq:problem} naturally. While the solution of each equation is not a hyperplane anymore, the projection to the solution set still has an explicit formula as follows. Let $\sA_i=\{\bz: y_i=|\ba_i^*\bz|\}$, then
\begin{equation}\label{eq:projection}
P_{\sA_i}(\bx)=\bx-\left(1-\frac{y_i}{|\ba_i^*\bx|}\right)	\bx^*\frac{\ba_i\ba_i^*}{\|\ba_i\|^2}.
\end{equation}
A randomized Kaczmarz update projects the estimate to the nearest point in $\sA_{r(k)}$ at the $k$-th iteration, where $r(k)$ is randomly chosen from $\{1, \cdots, m\}$, and the algorithm can be written as
\begin{equation}\label{eq:Kaczmarz}
\bx^{(k+1)}=P_{\sA_{r(k)}}(\bx^{(k)}).
\end{equation}

\subsection{Contribution and Main Result}

The main contribution of this paper is a guarantee on the linear convergence of randomized Kaczmarz algorithm as follows:
\begin{itemize}
\item First, this paper establishes a deterministic condition such that the algorithm converges linearly with high probability. Intuitively, the condition requires that an objective function is strongly convex in a neighborhood around the true signal $\bz$.
\item  Second, this paper proves that when the sensing vectors are sampled uniformly from a unit sphere in $\bbC^n$, and the number of sensing vectors $m$ satisfies $m>O(n\log n)$ as $m,n\rightarrow\infty$, the deterministic condition is satisfied with high probability. A similar result is also obtained for the unitary model,  where the sensing vectors are from the columns of random orthogonal matrices.
\end{itemize}
This paper generalizes the result in Tan and Vershynin in \cite{10.1093/imaiai/iay005} from the real-valued case to the complex-valued case. The generalization is not straightforward and the approach to obtain the deterministic condition is very different in this work. In comparison, since the phases can only be either $1$ or $-1$ in the real-valued case, \cite{10.1093/imaiai/iay005} divides all sensing vectors into ``good measurements'' with correct phases and ``bad measurements'' with possibly incorrect phases, and control the total influence of bad measurements. However, as remarked in \cite[Section 7.2]{10.1093/imaiai/iay005}, this method would not work in the complex-valued case since the phases are no longer $\pm1$ and each measurement contributes an error that scales with the phase difference, and we can no longer simply sum up the influence of bad measurements as in \cite[Lemma 2.1]{10.1093/imaiai/iay005}.

We remark that if $\ba_i$ are scaled such that $\|\ba_i\|=1$ for all $1\leq i\leq n$, then the update formula \eqref{eq:Kaczmarz} can also be considered as the stochastic gradient descent algorithm that minimizes $\frac{1}{m}\sum_{i=1}^m\Big(|\ba_i^*\bx|-y_i \Big)^2$, with step size chosen to be $1$. In this sense, our work is related to \cite{bassily2018exponential}, which the convergence of the stochastic gradient descent algorithm for generic objective functions has been studied, and both their work and this work are based on the convexity of the objective function. However, we remark that their result can not be directly applied here since it assumes a specific step size that depends on the smoothness constant and the Polyak-Lojasiewicz condition of the objective function, which is unclear for this objective function.

\subsection{Notation}
Throughout the paper, $C$ and $c$ are absolute constants that do not depend on $m, n$, and can change from line to line. We also implicitly assume that $m,n$ are sufficiently large, for example, we write $m>n+10$ when $m>n\log n$ is assumed. $\mathrm{Re}(x)$ represents the real component of a complex number $x$. For a set $\calS$, $|\calS|$ represents the cardinality of the set.

\section{Main Result}
We first present the main contribution of this paper, as well as a sketch of the proof and some discussion.
\begin{thm}\label{thm:main}
(a) Assuming that the sensing vectors $\{\ba_i\}_{i=1}^m$ are i.i.d. sampled from the uniform distribution on the unit sphere in $\bbC^n$, and in each iteration, the randomized Kaczmarz algorithm randomly picks up each equation with probability $1/m$. Then there exist  absolute constants $C_0, c_0, L$ that does not depend on $m,n$ such that if $m \geq C_0n\log n$ as $m,n\rightarrow\infty$, then for all $\|\bx^{(0)}-\bz\|\leq c_0\sqrt{\delta_1}\|\bz\|$ and $\epsilon>0$, we have \begin{equation}\label{eq:mainn}Pr\left(\|\bx^{(k)}-\bz\|^2\leq \epsilon \|\bx^{(0)}-\bz\|^2\right) \geq 1-\delta_1-\frac{\left(1-\frac{L}{n}\right)^k}{\epsilon\left(1-\delta_1\right)}-C\exp(-cn).  \end{equation}
(b) For the unitary model that $m=Kn$ for some integer $K$, and for any $1\leq k\leq K$, $[\ba_{(k-1)n+1},\cdots,\ba_{kn}]\in\bbC^{n\times n}$ is a random orthogonal matrix in $\bbC^{n\times n}$, there exists some constants $C_0, c_0, L$ such that if $m \geq C_0n\log n$ and $\sqrt{n}>\log^2 m$, then \eqref{eq:mainn} also holds as $m,n\rightarrow\infty$.
\end{thm}
This theorem shows the linear convergence of as follows: if $\delta_1<\frac{1}{2}$, and $k\geq \log(2\epsilon\delta_2)/\log(1-L/n)$, then with probability at least $1-\delta_1-\delta_2-C\exp(-cn)$, we have $\|\bx^{(k)}-\bz\|^2\leq \epsilon \|\bx^{(0)}-\bz\|^2$. If we let $\delta_1=\delta_2=\delta/2$, the number of iterations to achieve accuracy $\epsilon$ with probability $1-\delta$ is in the order of $O\left(n\log\frac{1}{\epsilon\delta}\right)$.

The proof of Theorem~\ref{thm:main} is divided into three steps. The first step establishes a condition in \eqref{eq:assumption}, with which the algorithm converges with high probability. This condition describes the regularity of an objective function in a local neighborhood around $\bz$. The second step establishes an explicit lower bound of  the key parameter $L$ in the condition \eqref{eq:assumption}. In the third step, we apply tools from random matrix theory and measure concentration to analyze the explicit formula of $L$ and show that it can be chosen as a constant as $n,m\rightarrow\infty$. The three steps  for part (a) are described in Sections~\ref{sec:step1}, ~\ref{sec:step2}, and ~\ref{sec:step3} respectively, and the third step for part (b) is described in Section~\ref{sec:step4}. The main results of these sections are summarized in Theorems~\ref{thm:main1} and ~\ref{thm:main2}, ~\ref{thm:main3}, and ~\ref{thm:step4}. Combining these theorems, we have the proof of Theorem~\ref{thm:main}.
\begin{proof}[Proof of Theorem~\ref{thm:main}]
WLOG we assume $\|\bz\|=1$ for the rest of the paper.  In the statement of Theorem~\ref{thm:main3}, $\alpha$ in \eqref{eq:main3} can be chosen such that $\alpha>1$ and $\frac{6.6}{\alpha-1}<\frac{c_1}{72}$ and $c_0$ in \eqref{eq:main3} can be chosen such that $(2+4\alpha)C4\sqrt{2}c_0\alpha<\frac{c_1}{72}$. Since when $\alpha$ is large, $c_0$ can be chosen in  the order of $1/\alpha^2$ and $c_0\alpha$ is in the order of $1/\alpha$, there exists a choice of $(\alpha,c_0)$ such that the assumption $2c_0\alpha<1$ in Theorem~\ref{thm:main3} holds. Then Theorem~\ref{thm:main2} and Theorem~\ref{thm:main3} imply that $L=\frac{c_1}{72}$ satisfies the assumption \eqref{eq:assumption}. With this assumption satisfied, Theorem~\ref{thm:main1} implies Theorem~\ref{thm:main}(a). The proof of Theorem~\ref{thm:main}(b) is similar to the proof of (a), with Theorem~\ref{thm:main3} replaced by Theorem~\ref{thm:step4}.
\end{proof}

\section{Step 1: Convergence under a deterministic condition}\label{sec:step1}
This section connects the convergence of the randomized Kaczmarz algorithm with the function
\begin{equation}\label{eq:obj}
f(\bx)=\frac{1}{m}\sum_{i=1}^m\Big(|\ba_i^*\bx|-y_i \Big)^2
\end{equation}and its directional derivatives defined by
\[
f'_{\bv}(\bx)=\lim_{t\rightarrow 0^+}\frac{f(\bx+t\bv)-f(\bv)}{t}= \frac{1}{m}\sum_{i=1}^m\Big(1-\frac{y_i}{|\ba_i^*\bx|}\Big)\big(\ba_i^*\bv\bx^*\ba_i+\ba_i^*\bx\bv^*\ba_i\big).
\]
The result of this section depends on the following local regularity assumption on $f$:
 \begin{equation}\label{eq:assumption}
 f(\bx)+f'_{\bz-\bx}(\bx)+\frac{L}{n}\|\bz-\bx\|^2 \leq f(\bz),\,\,\text{for all $\bx$ such that $\|\bx-\bz\|\leq c_0$.}
\end{equation}
We remark that this formulation is identical to the definition of strong convexity, so this assumption is related to the strong convexity of $f(\bx)$ in the local neighborhood $\bx\in B(\bz,c_0)$. However, it is slightly less restrictive in the sense that it only requires \eqref{eq:assumption} for a fixed $\bz$ (if $\bz$ is replaced by any $\by\in B(\bz,c_0)$, then this is equivalent to strong convexity).

We also remark that this objective function \eqref{eq:obj} has been studied in \cite{NIPS2016_6319,8049465,tan2019online} and its local regularity property has been studied in \cite{NIPS2016_6319,8049465}. However, these works study a different regularity assumption in~\cite[(12)]{NIPS2016_6319} and~\cite[(39)]{8049465}, which can be written as:
\[
  f'_{\bx-\bz}(\bx)\geq \frac{\mu}{2}\|f'(\bx)\|^2+\frac{\lambda}{2}\|\bx-\bz\|^2.
\]
In addition, similar to \cite{10.1093/imaiai/iay005}, these works only theoretically analyze the real-valued setting.

The main result of this section is summarized as follows. It states that under the assumption \eqref{eq:assumption}, the algorithm converges linearly with high probability.
\begin{thm}\label{thm:main1}
Assume that $\|\bz\|=1$ and $\|\ba_i\|=1$ for all $1\leq i\leq m$, \eqref{eq:assumption} holds, the randomized Kaczmarz algorithm randomly picks up an equation with the same probability $1/m$, and the algorithm is initialized such that $\|\bx^{(0)}-\bz\|\leq c_0\sqrt{\delta_1}$ for some $0\leq \delta_1\leq 1$. Then for any $\epsilon>0$,
\begin{equation}\label{eq:main1}
\Pr(\|\bx^{(k)}-\bz\|^2\leq \epsilon \|\bx^{(0)}-\bz\|^2) \geq 1-\delta_1-\frac{\left(1-\frac{L}{n}\right)^k}{\epsilon\left(1-\delta_1\right)}.
\end{equation}
\end{thm}
\begin{proof}
Let $P$ be the random mapping $P_{\sA_i}$ where $i$ is uniformly sampled from $\{1,\cdots, m\}$, apply the projection formula \eqref{eq:projection} with the assumption $\|\ba_i\|=1$ for all $1\leq i\leq m$, then
\begin{align*}
  &\Expect_{P}\|P\bx-\bz\|^2=\Expect_{i\sim\{1,\cdots,m\}}\Big\|\bx-\Big(1-\frac{y_i}{|\ba_i^*\bx|}\Big)	\bx^*\ba_i\ba_i^*-\bz\Big\|^2\\=&\Expect_{i}\left[\Big(1-\frac{y_i}{|\ba_i^*\bx|}\Big)^2|\ba_i^*\bx|^2-\mathrm{Re} \left(2\Big(1-\frac{y_i}{|\ba_i^*\bx|}\Big)\bx^*\ba_i\ba_i^*(\bx-\bz)\right)\right]+\|\bz-\bx\|^2\\
  =&\Expect_i\Big[({|\ba_i^*\bx|}-{y_i})^2\Big]+2 \Expect_i \left[\mathrm{Re}\left( \Big(1-\frac{y_i}{|\ba_i^*\bx|}\Big)\bx^*\ba_i\ba_i^*(\bz-\bx)\right)\right]+\|\bz-\bx\|^2
  \\
  =&f(\bx)+f'_{\bz-\bx}(\bx)+\|\bz-\bx\|^2\leq \Big(1-\frac{L}{n}\Big)\|\bz-\bx\|^2,
  \end{align*}
where the last inequality applies the assumption \eqref{eq:assumption} and $f(\bz)=0$, and the last equality use the fact that $\by+\by^*=2\mathrm{Re}(\by)$ (this is a fact that we will apply repetitively later).

The rest of the proof follows from the proof in \cite[Section 3]{10.1093/imaiai/iay005}. Let $\tau=\min\{k: \|\bx^{(k)}-\bz\|\leq c_0\},$ then following the proof in \cite[Theorem 3.1]{10.1093/imaiai/iay005}, we have $
P(\tau<\infty)\leq \left(\frac{c_0\sqrt{\delta_1}}{c_0}\right)^2=\delta_1.$ Following the proof in \cite[Corollary 3.2]{10.1093/imaiai/iay005}, \eqref{eq:main1} is proved.
\end{proof}

\section{The property of an implicit objective function}\label{sec:step2}
In this section, we will give an explicit formula for $L$ defined in \eqref{eq:assumption}. The formula will be based on a few additional definitions as follows. Let  $f_i(\bx)=({|\ba_i^*\bz|}-{|\ba_i^*\bx|})^2$, and define the first and the second directional derivatives of $f_i$ of direction $\bv$ at $\bx$ by
\begin{align}
f'_{i,\bv}(\bx)=\lim_{t\rightarrow 0}\frac{f_i(\bx+t\bv)-f_i(\bx)}{t},\\
f''_{i,\bv}(\bx)=\lim_{t\rightarrow 0}\frac{f'_i(\bx+t\bv)-f_i(\bx)}{t}.
\end{align}
It can be shown that the directional derivatives have explicit expressions
\begin{align*}
f'_{i,\bv}(\bx)&=\ba_i^*\bx\bv^*\ba_i+\ba_i^*\bv\bx^*\ba_i-|\ba_i^*\bz|\frac{\ba_i^*\bx\bv^*\ba_i+\ba_i^*\bv\bx^*\ba_i}{|\ba_i^*\bx|},\\
f''_{i,\bv}(\bx)&=2\ba_i^*\bv\bv^*\ba_i-|\ba_i^*\bz|\frac{2\ba_i^*\bv\bv^*\ba_i}{|\ba_i^*\bx|}+|\ba_i^*\bz|\frac{(\ba_i^*\bx\bv^*\ba_i+\ba_i^*\bv\bx^*\ba_i)^2}{2|\ba_i^*\bx|^3}.
\end{align*}
In addition, since $f(\bx)=\sum_{i=1}^mf_i(\bx)$, the first and the second directional derivative of $f(\bx)$ are $f'_{\bv}(\bx)=\sum_{i=1}^mf'_{i,\bv}(\bx)$ and $f''_{\bv}(\bx)=\sum_{i=1}^mf''_{i,\bv}(\bx)$.

\begin{thm}\label{thm:main2}
  For any $\bv, \bz\in\bbC^n$ with $\|\bv\|=\|\bz\|=1$ and $\beta>0$, define
  \begin{equation}\label{eq:define_S}
  \calS(\bv,\beta)=\{1\leq i\leq m: \beta |\ba_i^*\bv|\geq   |\ba_i^*\bz|\}
  \end{equation}
  then for any $\alpha>1$,
  \begin{align}\nonumber
  L=&\frac{n}{m}\min_{\|\bv\|=1} \left\{\frac{m}{2}f''_{\bv}(\bz)- \frac{6}{\alpha-1}\sum_{i=1}^m|\ba_i^*\bv|^2-(2+4\alpha)\sum_{i\in\calS(\bv,c_0\alpha)}|\ba_i^*\bv|^2\right\}\\
  =&\frac{n}{m}\min_{\|\bv\|=1} \left\{\!\!\frac{1}{2}\!\sum_{i=1}^m\!\frac{(\ba_i^*\bz\bv^*\!\ba_i\!\!+\!\!\ba_i^*\bv\bz^*\!\ba_i\!)^2}{2|\ba_i^*\bz|^2}\!-\! \frac{6}{\alpha\!-\!1}\!\!\sum_{i=1}^m|\ba_i^*\bv|^2\!-\!(2\!+\!4\alpha)\!\!\!\!\!\!\!\!\sum_{i\in\calS(\bv,c_0\alpha)}\!\!\!\!\!\!\!\!|\ba_i^*\bv|^2\!\!\right\}.\label{eq:define_L}
  \end{align}
  satisfies the assumption \eqref{eq:assumption}.
\end{thm}

\begin{proof}[Proof of Theorem~\ref{thm:main2}] We will first prove \eqref{eq:assumption} with $L$ defined in \eqref{eq:define_L} when $\|\bx-\bz\|= c_0$. For this case, there exists $\bv=\frac{\bx-\bz}{\|\bx-\bz\|}$ such that $\bx=\bz+ c_0\bv$ and $\|\bv\|=1$. For any $i\not\in\calS(\bv,c_0\alpha)$,  we have $|\ba_i^*\bz|\geq c_0\alpha |\ba_i^*\bv|$ and the triangle inequality implies
\[
\frac{|\ba_i^*(\bx-\bz)|}{|\ba_i^*\bz|} \leq \frac{1}{\alpha}, \,\,\,  \left|\frac{|\ba_i^*\bz|}{|\ba_i^*\bx|}-1\right|\leq \frac{1}{\alpha-1},\,\,\,\frac{|\ba_i^*\bz|}{|\ba_i^*\bx|}\leq \frac{\alpha}{\alpha-1}.
\]
Applying Lemma II.13 from \cite{zhang2020},
\[
\left|\frac{\ba_i^*\bx}{|\ba_i^*\bx|}-\frac{\ba_i^*\bz}{|\ba_i^*\bz|}\right|\leq 2\min\Big(\frac{|\ba_i^*(\bx-\bz)|}{|\ba_i^*\bz|},1\Big)=\frac{2}{\alpha}.
\]
Applying these inequalities, we have that for $i\not\in\calS(\bv,c_0\alpha)$,
\begin{align}
\nonumber&|f''_{i,\bv}(\bx)-f''_{i,\bv}(\bz)|\leq 2\left|\frac{|\ba_i^*\bz|}{|\ba_i^*\bx|}-1\right| \ba_i^*\bv\bv^*\ba_i+\left|\frac{|\ba_i^*\bz|}{|\ba_i^*\bx|}-1\right|\frac{(\ba_i^*\bz\bv^*\!\ba_i\!+\!\ba_i^*\bv\bz^*\!\ba_i)^2}{2|\ba_i^*\bz|^2}\\\nonumber&+ \frac{|\ba_i^*\bz|}{2|\ba_i^*\bx|} \Big[\big(\frac{\ba_i^*\bx}{|\ba_i^*\bx|}\bv^*\ba_i+\ba_i^*\bv\frac{\bx^*\ba_i}{|\bx^*\ba_i|}\big)^2-\big(\frac{\ba_i^*\bz}{|\ba_i^*\bz|}\bv^*\ba_i+\ba_i^*\bv\frac{\bz^*\ba_i}{|\bz^*\ba_i|}\big)^2\Big]\\\nonumber
\leq & \frac{2}{\alpha-1}\ba_i^*\bv\bv^*\ba_i+\frac{1}{\alpha-1}\frac{(\ba_i^*\bz\bv^*\ba_i+\ba_i^*\bv\bz^*\ba_i)^2}{2|\ba_i^*\bz|^2}\\\nonumber&+ \frac{|\ba_i^*\bz|}{2|\ba_i^*\bx|} \left[\big(\frac{\ba_i^*\bx}{|\ba_i^*\bx|}\bv^*\ba_i+\ba_i^*\bv\frac{\bx^*\ba_i}{|\bx^*\ba_i|}\big)-\big(\frac{\ba_i^*\bz}{|\ba_i^*\bz|}\bv^*\ba_i+\ba_i^*\bv\frac{\bz^*\ba_i}{|\bz^*\ba_i|}\big)\right]\\\nonumber&\left[\big(\frac{\ba_i^*\bx}{|\ba_i^*\bx|}\bv^*\ba_i+\ba_i^*\bv\frac{\bx^*\ba_i}{|\bx^*\ba_i|}\big)+\big(\frac{\ba_i^*\bz}{|\ba_i^*\bz|}\bv^*\ba_i+\ba_i^*\bv\frac{\bz^*\ba_i}{|\bz^*\ba_i|}\big)\right]\\\leq
&\frac{2}{\alpha-1}|\ba_i^*\bv|^2+\frac{2}{\alpha-1}|\ba_i^*\bv|^2+\frac{\alpha}{2(\alpha-1)}\left[\frac{4}{\alpha}|\ba_i^*\bv|\right]\Big[4|\ba_i^*\bv|\Big]\leq \frac{12}{\alpha-1}|\ba_i^*\bv|^2,\label{eq:S1}
\end{align}
where the intermediate inequalities use the fact $|\ba_i^*\bz\bv^*\ba_i|\leq |\ba_i^*\bz||\ba_i^*\bv|.$

For any $i\in\calS(\bv,c_0\alpha)$ and any $0\leq t\leq c_0$, we have
\[\left|\frac{\ba_i^*(\bz+t\bv)}{|\ba_i^*(\bz+t\bv)|}-\frac{\ba_i^*(\bz+c_0\bv)}{|\ba_i^*(\bz+c_0\bv)|}\right|\leq 2\]
and
\begin{align}\nonumber
&|f_{i,\bv}'(\bz+t\bv)-f_{i,\bv}'(\bz+c_0\bv)|\nonumber\\=& 2(t-c_0)|\ba_i^*\bv|^2-2|\ba_i^*\bz|\mathrm{Re}\left(\left(\frac{\ba_i^*(\bz+t\bv)}{|\ba_i^*(\bz+t\bv)|}-\frac{\ba_i^*(\bz+c_0\bv)}{|\ba_i^*(\bz+c_0\bv)|}\right)\bv\ba_i^*\right)\nonumber\\\leq&2(c_0-t)|\ba_i^*\bv|^2+4 |\ba_i^*\bz||\ba_i^*\bv|\leq (2(c_0-t)+4c_0\alpha)|\ba_i^*\bv|^2.\label{eq:S2}
\end{align}

Combining the two cases studied in \eqref{eq:S1} and \eqref{eq:S2}, we have\begin{align*}
&m [ f(\bx)-f(\bz)+f'_{\bx-\bz}(\bx)]=m[f(\bz+c_0\bv)-f(\bz)-c_0f'_{\bv}(\bz+c_0\bv)]\\=&m \int_{t=0}^{c_0}\Big[f'_{\bv}(\bz+t\bv)-f'_{\bv}(\bz+c_0\bv)\Big]\di t\\=&
\sum_{i\not\in\calS(\bv,c_0\alpha)}\int_{t=0}^{c_0}\Big[f_{i,\bv}'(\bz+t\bv)-f_{i,\bv}'(\bz+c_0\bv)\Big]\di t\\&+\sum_{i\in\calS(\bv,c_0\alpha)}\int_{t=0}^{c_0}\Big[f_{i,\bv}'(\bz+t\bv)-f_{i,\bv}'(\bz+c_0\bv)\Big]\di t\\
= & \!-\!\!\!\!\!\!\sum_{i\not\in\calS(\bv,c_0\alpha)}\!\!\!\!\!\int_{t=0}^{c_0}\!\!\int_{a=t}^{c_0}\!\! f_{i,\bv}''(\bz+a\bv) \di a \di t+\!\!\!\!\!\!\!\sum_{i\in\calS(\bv,c_0\alpha)}\!\!\!\!\!\int_{t=0}^{c_0}\!\!\Big[f_{i,\bv}'(\bz+t\bv)-f_{i,\bv}'(\bz+c_0\bv)\Big]\di t
\\\leq & - \sum_{i\not\in\calS(\bv,c_0\alpha)}\int_{t=0}^{c_0}\int_{a=t}^{c_0}\Big[f''_{i,\bv}(\bz)-\frac{12}{\alpha-1}|\ba_i^*\bv|^2\Big]\di a \di t\\&+\sum_{i\in\calS(\bv,c_0\alpha)}\int_{t=0}^{c_0}\Big[(2(b-t)+4c_0\alpha)|\ba_i^*\bv|^2\Big]\di t
\\=&c_0^2\left\{- \sum_{i\not\in\calS(\bv,c_0\alpha)} \Big[\frac{f''_{i,\bv}(\bz)}{2}-\frac{6}{\alpha-1}|\ba_i^*\bv|^2\Big]+\sum_{i\in\calS(\bv,c_0\alpha)}(1+4\alpha)|\ba_i^*\bv|^2 \right\}
\\\leq &c_0^2\left\{ - \sum_{i=1}^m \Big[\frac{f''_{i,\bv}(\bz)}{2}-\frac{6}{\alpha-1}|\ba_i^*\bv|^2\Big]\!+\!\!\!\!\!\!\sum_{i\in\calS(\bv,c_0\alpha)}\frac{f''_{i,\bv}(\bz)}{2}+\!\!\!\!\!\sum_{i\in\calS(\bv,c_0\alpha)}\!\!\!\!\!(2+4\alpha)|\ba_i^*\bv|^2 \right\}\\
\leq &{\|\bx-\bz\|^2}\left\{ - \frac{m}{2}f''_{\bv}(\bz) + \sum_{i=1}^m \frac{6}{\alpha-1}|\ba_i^*\bv|^2+\sum_{i\in\calS(\bv,c_0\alpha)}(2+4\alpha)|\ba_i^*\bv|^2 \right\},
\end{align*}
where the first inequality follows from \eqref{eq:S1}  and \eqref{eq:S2}, and  the last inequality applies the observation that $|f''_{i,\bv}(\bz)|=|\frac{(\ba_i^*\bx\bv^*\ba_i+\ba_i^*\bv\bx^*\ba_i)^2}{2|\ba_i^*\bx|^2}|\leq 2|\ba_i^*\bv|^2$. Recall the definition of $L$ in \eqref{eq:assumption}, we have proved \eqref{eq:assumption} with $L$ defined in \eqref{eq:define_L} when $\|\bx-\bz\|= c_0$.

When $\|\bx-\bz\|< c_0$, applying the same procedure we can show that \eqref{eq:assumption} holds with $L$ defined by
\begin{equation}\label{eq:define_L1}
\frac{n}{m}\min_{\|\bv\|=1} \left\{\frac{m}{2}f''_{\bv}(\bz)- \frac{6}{\alpha-1}\sum_{i=1}^m|\ba_i^*\bv|^2-(2+4\alpha)\sum_{i\in\calS(\bv,\|\bx-\bz\|\alpha)}|\ba_i^*\bv|^2\right\}.
\end{equation}
By definition, $\|\bx-\bz\|< c_0$ implies that $\calS(\bv,\|\bx-\bz\|\alpha)\subseteq \calS(\bv,c_0\alpha)$. As a result, the expression in \eqref{eq:define_L1} is greater or equal than $L$ defined in \eqref{eq:define_L}, and \eqref{eq:assumption} with $L$ defined in \eqref{eq:define_L} also holds when $\|\bx-\bz\| < c_0$, and Theorem~\ref{thm:main2} is proved.
\end{proof}

\section{The Gaussian model}\label{sec:step3}
In this section, we will show that under the Gaussian model, $L$ defined in \eqref{eq:assumption} and \eqref{eq:define_L} can be estimated and has a lower bound.
\begin{thm}\label{thm:main3}
Assuming that the  sensing vectors $\{\ba_i\}_{i=1}^m$ are selected uniformly and independently from the unit sphere in $\bbC^n$ and $2c_0\alpha>1$, then there exists $C_0$ that does not depend on $n$ and $m$ such that when $m>C_0n\log n$,  with probability at least $1-Cn\exp(-cn)$, $L$ defined in \eqref{eq:define_L} has a lower bound
  \begin{equation}\label{eq:main3}
L\geq \frac{c_1}{24}-\frac{6.6}{\alpha-1}-(2+4\alpha)C4\sqrt{2}c_0\alpha.
  \end{equation}
\end{thm}
\begin{proof}
The proof is based on bounding each component in \eqref{eq:define_L} separately in  Lemma~\ref{lemma:term2}, Lemma~\ref{lemma:term1}, and Lemma~\ref{lemma:term3}. Combining these estimations with $\delta=\exp(-n\log n)$ in Lemma~\ref{lemma:term2} and $\beta=2c_0\alpha$ in Lemma~\ref{lemma:term3}, Theorem~\ref{thm:main3} is proved.\end{proof}
The following is a restatement of~\cite[Lemma 5.2]{vershynin2010introduction}. While the original setting is real-valued, it is easy to generalize it to the complex-valued setting by treating $\bbC^n$ as $\reals^{2n}$.
\begin{lemma}[Covering numbers of the sphere]\label{lemma:covering}
  There exists an $\epsilon$-net over the unit sphere in $\bbC^{n}$ equipped with the Euclidean metric, with at most $(1+\frac{2}{\epsilon})^{2n}$ points.
\end{lemma}
The following lemma bounds the second term in \eqref{eq:define_L}. In fact, this term is the squared operator norm of $\bA$ and has been well studied, and the following result follows from \cite[Lemma 5.8]{10.1093/imaiai/iay005}.
\begin{lemma}[The bound on the second term in \eqref{eq:define_L}]\label{lemma:term2}
 If $m\geq C(n+\sqrt{\log(1/\delta)})$, then for any $\delta>0$,
\[
\Pr\left(\max_{\|\bv\|=1}\frac{1}{m}\sum_{i=1}^m\|\ba_i^*\bv\|^2\leq \frac{1.1}{n}\right)\geq 1-\delta.
\]
\end{lemma}
The following lemma bounds the first term in \eqref{eq:define_L}. The proof is based on an $\epsilon$-net argument. The proof is deferred to Section~\ref{sec:term1}.
\begin{lemma}[The bound on the first term in \eqref{eq:define_L}]\label{lemma:term1}
For any fixed $\bz, \bv\in\mathbb{C}^n$ with $\|\bz\|=\|\bv\|=1$, assuming that $\{\ba_i\}_{i=1}^m$ are selected uniformly and independently from the unit sphere in $\mathbb{C}^n$, there exists $c_1>0$ such that
\begin{align}\nonumber
  & \Pr\left(\min_{\|\bv\|=1}\frac{1}{m}\sum_{i=1}^m\frac{(\ba_i^*\bz\bv^*\ba_i+\ba_i^*\bv\bz^*\ba_i)^2}{2|\ba_i^*\bz|^2}\geq \frac{c_1m}{24n}\right)\\\geq& 1-\left(\exp\left(-\frac{m}{576}\right)+\exp\left(-\frac{m}{8}\right)\right)\left(1+\frac{192n}{c_1}\right)^{2n}.  \label{eq:term1_5}
  \end{align}
\end{lemma}
The following lemma bounds the last term in \eqref{eq:define_L}. We remark that it shares some similarities with the estimation in~\cite[Theorem 5.7]{10.1093/imaiai/iay005}, in the sense that both estimations depend on a ``wedge''. However, the ``wedge'' $\calS(\bv,\beta/2)$ in the complex setting is more complicated and the argument based on VC theory in \cite{10.1093/imaiai/iay005} does not apply. Instead, the proof of Lemma~\ref{lemma:term3} consists of two parts: first, we control the size of $\calS(\bv,\beta/2)$ based on an $\epsilon$-net argument. Then we can invoke the metric entropy/chaining argument in \cite{10.1093/imaiai/iay005}. The proof is deferred to Section~\ref{sec:term3}.
\begin{lemma}[The bound on the third term in \eqref{eq:define_L}]\label{lemma:term3}
  Assume $\beta>1$ and  $m\geq \max(n,\frac{n\log n}{2\beta^2})$, then there exists $C$ that does not depend on $n, m, \beta$ such that
  \[
\Pr\!\left(\! \max_{\|\bv\|=1} \! \!\!\!\sum_{i\in\calS(\bv,\beta/2)}\!\!\!\!|\ba_i^*\bv|^2 \leq
  C \sqrt{2}\beta\frac{m}{n}\!\right)\!\!\geq 1-2n\exp(-n)-2\exp\!\Big(\!-\frac{m\beta^4}{2}\!\Big)(1+2n)^{2n}.
  \]
  \end{lemma}
 \subsection{Proof of Lemma~\ref{lemma:term1}} \label{sec:term1}
\begin{proof}
The proof will be based on an $\epsilon$-net argument. In the first step, we will show that for any $\bv$ with $\|\bv\|=1$, the event in the LHS of \eqref{eq:term1_5} happens with high probability. Second, we will establish a perturbation bound on $\bv$ when the perturbation is smaller than $\epsilon$. Then a standard $\epsilon$-net argument will be applied.

First, we note that
\begin{align}\label{eq:term1_0}
  &\frac{(\ba_i^*\bz\bv^*\ba_i+\ba_i^*\bv\bz^*\ba_i)^2}{2|\ba_i^*\bz|^2}=  \frac{(\ba_i^*\bz\bv^*\ba_i+\ba_i^*\bv\bz^*\ba_i)^2}{2|\ba_i^*\bz|^2\|P_{\Sp(\bv,\bz)}(\ba_i)\|^2}  \|P_{\Sp(\bv,\bz)}(\ba_i)\|^2
  \\= &\frac{(\bb_i^*\hat{\bz}\hat{\bv}^*\bb_i+\bb_i^*\hat{\bv}\hat{\bz}^*\bb_i)^2}{2|\bb_i^*\hat{\bz}|^2} \|P_{\Sp(\bv,\bz)}(\ba_i)\|^2,\nonumber
\end{align}
where $\bb_i, \hat{\bz}, \hat{\bv}\in\bbC^2$ are defined by $\bb_i=\frac{P_{\Sp(\bv,\bz)}(\ba_i)}{\|P_{\Sp(\bv,\bz)}(\ba_i)\|}$, $\hat{\bz}=P_{\Sp(\bv,\bz)}\bz$, and $\hat{\bv}=P_{\Sp(\bv,\bz)}\bv$. Here $\bb_i$ is sampled uniformly and independently from a unit circle in $\bbC^2$, and it is independent of $\|P_{\Sp(\bv,\bz)}(\ba_i)\|$.

Considering that $P_{\Sp(\bv,\bz)}(\ba_i)$ is a projection of a random unit vector from $\bbC^n$ to $\bbC^2$, there exists $c_1>0$ such that (in fact, one can verify it with $c_1=0.8$)
\[
\Pr\left(\|P_{\Sp(\bv,\bz)}(\ba_i)\|^2\geq \frac{c_1}{n}\right)\geq \frac{3}{4}.
\]
Then Hoeffding's inequality suggests that
\begin{equation}\label{eq:term1_1}
P\left(\sum_{i=1}^mI\left( \|P_{\Sp(\bv,\bz)}(\ba_i)\|^2\geq \frac{c_1}{n}\right)\geq \frac{m}{2}\right)\geq 1-\exp\Big(-\frac{m}{8}\Big).
\end{equation}
For the component $\frac{(\bb_i^*\hat{\bz}\hat{\bv}^*\bb_i+\bb_i^*\hat{\bv}\hat{\bz}^*\bb_i)^2}{2|\bb_i^*\hat{\bz}|^2}$, note that $\bb_i$ is distributed uniformly on a circle in $\bbC^2$, so WLOG we may assume that $\hat{\bz}=[1,0]$ and $\hat{\bv}=[\cos \theta, \sin\theta]$, and then
\begin{align*}
 & \bb_i^*\hat{\bz}\hat{\bv}^*\bb_i+\bb_i^*\hat{\bv}\hat{\bz}^*\bb_i=2\mathrm{Re}(  \bb_i^*\hat{\bz}\hat{\bv}^*\bb_i)=2\mathrm{Re}(\cos\theta|\bb_{i,1}|^2+\sin\theta\bb_{i,1}^*\bb_{i,2})\\=&2\cos\theta|\bb_{i,1}|^2 +2\cos\theta\sin\theta \mathrm{Re}(\bb_{i,1}^*\bb_{i,2}),
\end{align*}
and applying $|\bb_i^*\hat{\bz}|^2=|\bb_{i,1}|^2$,
\begin{equation}\label{eq:temp1}
  \frac{(\bb_i^*\hat{\bz}\hat{\bv}^*\bb_i+\bb_i^*\hat{\bv}\hat{\bz}^*\bb_i)^2}{2|\bb_i^*\hat{\bz}|^2}=\cos^2\theta |\bb_{i,1}|^2+2\cos\theta\sin\theta  \mathrm{Re}(\bb_{i,1}^*\bb_{i,2})+\sin^2\theta \frac{(\mathrm{Re}(\bb_{i,1}^*\bb_{i,2}))^2}{|\bb_{i,1}|^2}.
\end{equation}
By the symmetry of the distribution of $\bb_{i,2}$, when $\bb_{i,1}$ is fixed, $\Expect_{\bb_{i,2}} {\mathrm{Re}(\bb_{i,1}^*\bb_{i,2})}=0$ and $\Expect_{\bb_{i,2}}\frac{(\mathrm{Re}(\bb_{i,1}^*\bb_{i,2}))^2}{|\bb_{i,1}|^2}=\frac{1}{2}\Expect_{\bb_{i,2}} \frac{|\bb_{i,1}^*\bb_{i,2}|^2}{|\bb_{i,1}|^2}=\frac{1}{2}\Expect_{\bb_{i,2}}|\bb_{i,2}|^2=\frac{1}{2}(1-|\bb_{i,1}|^2)$. Combining it with $\Expect |\bb_{i,1}|^2=1/2$ and \eqref{eq:temp1}, we have
\[
  \Expect\left[ \frac{(\bb_i^*\hat{\bz}\hat{\bv}^*\bb_i+\bb_i^*\hat{\bv}\hat{\bz}^*\bb_i)^2}{2|\bb_i^*\hat{\bz}|^2}\right]=\frac{1}{2}\cos^2\theta+\frac{1}{4}\sin^2\theta\geq \frac{1}{4}.
\]
Since for each $1\leq i\leq m$,
\[
 \frac{(\bb_i^*\hat{\bz}\hat{\bv}^*\bb_i+\bb_i^*\hat{\bv}\hat{\bz}^*\bb_i)^2}{2|\bb_i^*\hat{\bz}|^2}\leq 2|\bb_i^*\hat{\bv}|^2\leq 2,
\]
When the event in \eqref{eq:term1_1} holds,
\begin{equation}\label{eq:temp2}
\text{there exists a set $\sI$ of size $m/2$ such that for all $i\in\sI$, $\|P_{\Sp(\bv,\bz)}(\ba_i)\|^2\geq \frac{c_1}{n}$.}\end{equation} Hoeffding's inequality then gives the estimation
\begin{equation}\label{eq:term1_2}
\Pr\left(\frac{2}{m}\sum_{i\in\sI}\frac{(\bb_i^*\bz\bv^*\bb_i+\bb_i^*\bv\bz^*\bb_i)^2}{2|\bb_i^*\bz|^2}\geq \frac{1}{4}-t\right)\geq 1-\exp\left(-\frac{mt^2}{4}\right).
\end{equation}
Combining \eqref{eq:term1_0}, \eqref{eq:temp2}, and \eqref{eq:term1_2},
\begin{equation}\label{eq:term1_3}
  \Pr\left(\frac{2}{m}\sum_{i\in\sI}\frac{(\ba_i^*\bz\bv^*\ba_i+\ba_i^*\bv\bz^*\ba_i)^2}{2|\ba_i^*\bz|^2}\geq \Big(\frac{1}{4}-t\Big)\frac{c_1}{n}\right)\geq 1-\exp\left(-\frac{mt^2}{4}\right)-\exp\Big(-\frac{m}{8}\Big),
  \end{equation}
which leads to
\begin{equation}\label{eq:term1_4}
  \Pr\left(\sum_{i=1}^m\frac{(\ba_i^*\bz\bv^*\ba_i+\ba_i^*\bv\bz^*\ba_i)^2}{2|\ba_i^*\bz|^2}\geq \Big(\frac{1}{4}-t\Big)\frac{c_1m}{2n}\right)\geq  1- \exp\left(-\frac{mt^2}{4}\right)-\exp\Big(-\frac{m}{8}\Big).
\end{equation}

Second, we will establish a perturbation bound on $\bv$. For any $\bv'$ such that $\|\bv'-\bv\|\leq \epsilon$,
\begin{align*}
&  \sum_{i=1}^m\frac{(\ba_i^*\bz\bv^*\ba_i+\ba_i^*\bv\bz^*\ba_i)^2}{2|\ba_i^*\bz|^2}-\sum_{i=1}^m\frac{(\ba_i^*\bz\bv'^*\ba_i+\ba_i^*\bv'\bz^*\ba_i)^2}{2|\ba_i^*\bz|^2}\\=&4\sum_{i=1}^m\frac{\mathrm{Re}(\ba_i^*\bz(\bv+\bv')^*\ba_i)\mathrm{Re}(\ba_i^*\bz(\bv-\bv')^*\ba_i)}{2|\ba_i^*\bz|^2}\leq 4\epsilon\sum_{i=1}^m\|\ba_i\|^2=4m\epsilon.
\end{align*}

Combining it with \eqref{eq:term1_4}, Lemma~\ref{lemma:covering}, and a standard $\epsilon$-net argument, we have
\begin{align*}\nonumber
  &  \Pr\left(\min_{\|\bv\|=1}\frac{1}{m}\sum_{i=1}^m\frac{(\ba_i^*\bz\bv^*\ba_i+\ba_i^*\bv\bz^*\ba_i)^2}{2|\ba_i^*\bz|^2}\geq \frac{c_1m}{2n}\Big(\frac{1}{4}-t\Big)-4m\epsilon\right)\\\geq& 1-\left(\exp\left(-\frac{mt^2}{4}\right)+\exp\left(-\frac{m}{8}\right)\right)\left(1+\frac{2}{\epsilon}\right)^{2n}.
  \end{align*}
Set $t=\frac{1}{12}$ and $\epsilon=\frac{c_1}{96n}$, \eqref{eq:term1_5} and Lemma~\ref{lemma:term1} are proved.
\end{proof}

\subsection{Proof of Lemma~\ref{lemma:term3}} \label{sec:term3}
In this section, we will first present a few lemmas and their proof, and then prove Lemma~\ref{lemma:term3} in the end. In particular, we will estimate the expected value of $|\calS(\bv,\beta)|$ in Lemma~\ref{lemma:term34}; then we will investigate the perturbation of $|\calS(\bv,\beta)|$ when $\bv$ is perturbed in Lemma~\ref{lemma:term31}. Next, an $\epsilon$-net argument will be used to give a uniform upper bound on $|\calS(\bv,\beta)|$ for all $\bv$. Combining this uniform upper bound and Lemma~\ref{lemma:term32}, Lemma~\ref{lemma:term3} is proved.

\begin{lemma}\label{lemma:term34}
Fix $\bv$ with $\|\bv\|=1$ and assume $\beta\geq 1$, then for each $1\leq i\leq m$,  \[
   \Pr\{|\ba_i^*\bz|\leq \beta |\ba_i^*\bv|\}\leq \frac{\beta^2}{1+\beta^2}.
  \]
\end{lemma}
\begin{proof}[Proof of Lemma~\ref{lemma:term34}]
First, we will show that it is sufficient to prove the case $\bv\perp\bz$. Assume $\bv\not\perp\bz$, then there exists $\bu$ such that $\|\bu\|=1$, $\bu\perp\bz$, $\bv\in\Sp(\bu,\bz)$. WLOG assume that $\bv=e^{i\eta_1}\cos\theta\bz+e^{i\eta_2}\sin\theta\bu$, then we have $|\ba_i^*\bv|^2=\cos^2\theta|\ba_i^*\bz|^2+\sin^2\theta|\ba_i^*\bu|^2$ and
\begin{equation}\label{eq:buv}
\frac{|\ba_i^*\bv|^2}{|\ba_i^*\bz|^2} = \frac{\cos^2\theta|\ba_i^*\bz|^2+\sin^2\theta|\ba_i^*\bu|^2}{|\ba_i^*\bz|^2}=\cos^2\theta+\sin^2\theta\frac{|\ba_i^*\bu|^2}{|\ba_i^*\bz|^2}.
\end{equation}
Assuming that  $i\in \calS(\bv,\beta)$, then $\frac{|\ba_i^*\bv|^2}{|\ba_i^*\bz|^2}\geq \frac{1}{\beta^2}\geq 1$ and \eqref{eq:buv} implies
\[
  \frac{|\ba_i^*\bu|^2}{|\ba_i^*\bz|^2}\geq \frac{|\ba_i^*\bv|^2}{|\ba_i^*\bz|^2},
\]
so $i$ is also in $\calS(\bu,\beta)$. Therefore, $\calS(\bv,\beta)\subseteq \calS(\bu,\beta)$, and it is sufficient to prove Lemma~\ref{lemma:term34} for the case $\bv\perp\bz$.

Note that $\Pr\{|\ba_i^*\bz|\leq \beta |\ba_i^*\bv|\}$ does not change with the scaling of $\ba_i$, so we may assume that $\ba_i$ are sampled from complex Gaussian distribution $CN(0,\sqrt{2}\bI)$, i.e., each real and imaginary component are sampled from $N(0,1)$. With $\bv\perp\bz$,
\[\Pr\left(\ba_i^*\bz|\leq \beta |\ba_i^*\bv|\right)=\Pr\left(\sqrt{g_1^2+g_2^2}\leq \beta \sqrt{g_3^2+g_4^2}\right),\] assuming that $g_1,g_2,g_3,g_4$ are i.i.d. sampled from $N(0,1)$.  By calculation, both  $\sqrt{g_1^2+g_2^2}$ and $\sqrt{g_3^2+g_4^2}$ have the probability density function $xe^{-x^2/2}$ with $x\geq 0$. Therefore,
\begin{align*}
&\Pr\left(\sqrt{g_1^2+g_2^2}\leq \beta \sqrt{g_1^2+g_2^2}\right)=\int_{x=0}^\infty xe^{-x^2/2}\int_{y=\frac{x}{\beta}}^\infty ye^{-y^2/2}\di y\di x \\=& \int_{x=0}^\infty xe^{-x^2/2}e^{-x^2/2\beta^2}\di x=\frac{\beta^2}{1+\beta^2},
\end{align*}
and Lemma~\ref{lemma:term34} is then proved.
\end{proof}

\begin{lemma}\label{lemma:term31}
For any $\bv, \bv'\in\bbC^n$ with $\|\bv-\bv'\|\leq\epsilon$, we have
\[
\Big|\calS\Big(\bv',\frac{\beta}{1-\epsilon c_2}\Big)\Big|\leq |\calS(\bv,\beta)| + |\{1\leq i\leq m: \|\ba_i\|/|\ba_i^*\bv|\geq c_2\}|.
\]
\end{lemma}
\begin{proof}[Proof of Lemma~\ref{lemma:term31}]
If $i\in\calS(\bv,\beta)$ and $\|\ba_i^*\|/|\ba_i^*\bv|\leq c_2$, then
\[
|\ba_i^*\bv|\leq  |\ba_i^*\bv'|+\epsilon\|\ba_i^*\|\leq  |\ba_i^*\bv'|+\epsilon c_2|\ba_i^*\bv|,
\]
so $\|\ba_i^*\bz\|\leq \beta|\ba_i^*\bv|\leq
\frac{\beta}{1-\epsilon c_2}|\ba_i^*\bv'|$.
\end{proof}

The following is a restatement of \cite[Theorem 5.7]{10.1093/imaiai/iay005}. While the statement is proved for the real-valued case, it can be generalized to the complex-valued case by treating the real component and the imaginary component separately.
\begin{lemma}\label{lemma:term32}
Let $0<\delta<1/2$, $0<c_3<1$, suppose $m\geq \max(n,\log(1/\delta)/c_3)$, then with probability at least $1-2\delta$, for any set $\calS\in\{1,\cdots,m\}$ with $|\calS|\leq c_3 m$, we have
\[
\Big\|\sum_{i\in\calS}\ba_i\ba_i^*\Big\|  \leq C \sqrt{c_3}\frac{m}{n}.
\]
\end{lemma}

\begin{proof}[Proof of Lemma~\ref{lemma:term3}]
  To investigate the probability that $\|\ba_i^*\|/|\ba_i^*\bv|\geq c_2$, WLOG we may assume scale $\ba_i\in\bbC^n$ and assume that the real and the complex component of each element of $\ba_i$ is sampled from $N(0,1)$. Then  $|\ba_i^*\bv|$ has a p.d.f. of $f(x)=xe^{-x^2/2}$ for $x\geq 0$, and \begin{equation}\label{eq:temp4}
  \Pr(|\ba_i^*\bv|\geq t)=\exp(-t^2/2).
  \end{equation}
  In addition, $\|\ba_i^*\|^2$ has the same distribution of $\chi^2$ distribution with a degree of freedom $2n$. As a result, the tail bound of the $\chi^2$ distribution~\cite[Example 2.11]{wainwright2019high} implies (using one-sided inequality, apply $t=1$ and replace $n$ by $2n$ from their notation)
\begin{equation}\label{eq:temp5}
  \Pr(\|\ba_i\|^2\geq 4n)\leq \exp(-n/4).
  \end{equation}
  Combining \eqref{eq:temp4} with $t=\frac{2\sqrt{n}}{c_2}$ and \eqref{eq:temp5}, we have
  \[
  \Pr(\|\ba_i^*\|/|\ba_i^*\bv|\geq c_2 )\leq 1-\exp(-2n/c_2^2)+\exp(-n/4).
  \]
  Combining it with Lemma~\ref{lemma:term34}, we have that for each $1\leq i\leq m$,
  \[
  \Pr\left(\frac{\|\ba_i\|}{|\ba_i^*\bv|}\geq c_2\right)\leq    1-\exp\Big(-\frac{2n}{c_2^2}\Big)+\exp\Big(-\frac{n}{4}\Big).
  \]
  Applying Hoeffding inequality, for each $\bv$,
  \begin{align}\nonumber
  &\Pr\left(\Big|\Big\{1\leq i\leq m: \frac{\|\ba_i\|}{|\ba_i^*\bv|}\geq c_2\Big\}\Big|\geq m\Big(1-\exp\Big(-\frac{2n}{c_2^2})+\exp\Big(-\frac{n}{4}\Big) +t\Big)\right)\\\leq &\exp(-2mt^2).  \label{eq:temp6}
  \end{align}
  Similarly, applying Hoeffding inequality to Lemma~\ref{lemma:term34}, for each $\bv$,
  \begin{align}\label{eq:temp7}
  &\Pr\left(\Big|\calS(\bv,\beta)\Big\}\Big|\geq m\Big(\frac{\beta^2}{1+\beta^2} +t\Big)\right)\leq \exp(-2mt^2).
  \end{align}
Combining \eqref{eq:temp6}, \eqref{eq:temp7} with $c_2=1/2\epsilon$, then the standard $\epsilon$-net argument, Lemma~\ref{lemma:covering}, and Lemma~\ref{lemma:term31} imply
  \begin{align}\nonumber
 & \Pr\left(\max_{\|\bv\|=1}\Big|\calS\Big(\bv,\frac{\beta}{2}\Big)\Big|\geq m\Big(1-\exp(-8n\epsilon^2)+\exp\big(-\frac{n}{4}\big) +\frac{\beta^2}{1+\beta^2}+2t\Big)\right)\\\leq& 2\exp(-2mt^2)(1+2/\epsilon)^{2n}.  \label{eq:temp9}
  \end{align}
Let $\epsilon=1/n$, $\eta=\exp(-n\log n)$, $t=\beta^2/2$, $c_3=2\beta^2$, note that $1-\exp(-8n\epsilon^2)+\exp\big(-\frac{n}{4}\big)\rightarrow 0$ as $n\rightarrow\infty$, Lemma~\ref{lemma:term3} is then proved using Lemma~\ref{lemma:term32} and the fact that $\Big\|\sum_{i\in\calS}\ba_i\ba_i^*\Big\|\geq \sum_{i\in\calS}|\ba_i^*\bv|^2$.
  \end{proof}

\section{The unitary model}\label{sec:step4}
In this section, we will show that under the unitary model, $L$ defined in \eqref{eq:assumption} and \eqref{eq:define_L} can be estimated and has a lower bound as follows':
\begin{thm}\label{thm:step4}
Under the unitary model, there exists $c_0$ and $\alpha$ such that $2c_0\alpha<1$, and
\begin{equation}\label{eq:unitary1}
L=\frac{n}{m}\min_{\|\bv\|=1} \left\{\!\!\frac{1}{2}\!\sum_{i=1}^m\!\frac{(\ba_i^*\bz\bv^*\!\ba_i\!\!+\!\!\ba_i^*\bv\bz^*\!\ba_i\!)^2}{2|\ba_i^*\bz|^2}\!-\! \frac{6}{\alpha\!-\!1}\!\!\sum_{i=1}^m|\ba_i^*\bv|^2\!-\!(2\!+\!4\alpha)\!\!\!\!\!\!\!\!\sum_{i\in\calS(\bv,c_0\alpha)}\!\!\!\!\!\!\!\!|\ba_i^*\bv|^2\!\!\right\}.
\end{equation}
is bounded from below with high probability, that is, $L>c$ for some $c>0$ with probability $1-Cn\exp(-Cn)$.
\end{thm}
We will control each of the three expressions in \eqref{eq:unitary1} separately. For the component $\frac{6}{\alpha\!-\!1}\!\!\sum_{i=1}^m|\ba_i^*\bv|^2$, by definition it is equivalent to $\frac{6}{\alpha-1}K=\frac{6}{\alpha-1}\frac{m}{n}$. Assuming that we have that for some $c_4>0$,
\begin{equation}\label{eq:unitary2}
\Pr\left\{\min_{\|\bv\|=1}\!\!\frac{(\ba_i^*\bz\bv^*\!\ba_i\!\!+\!\!\ba_i^*\bv\bz^*\!\ba_i\!)^2}{2|\ba_i^*\bz|^2}\leq \frac{c_4}{2}\frac{m}{n}\!\!\right\}\geq 1-C\exp(-Cm+n\log n).
\end{equation}
and that there exists $c_8>0$ such that for any $\beta$,
\begin{equation}\label{eq:unitary3}
\Pr\left\{\max_{\|\bv\|=1}\sum_{i\in\calS(\bv,\beta)}\!\!|\ba_i^*\bv|^2\leq c_8 \beta \frac{m}{n}\!\!\right\}\geq 1-C\exp(-Cm+n\log n),
\end{equation}
then we can choose $\alpha$ and $c_0$ such that $2c_0\alpha<1$ and $L>c$ for some $c>0$ with probability $1-Cn\exp(-Cn)$. That is, Theorem~\ref{thm:step4} is proved.

\subsection{Proof of \eqref{eq:unitary2}}
 For the component $\min_{\|\bv\|=1}\sum_{i=1}^m\!\frac{(\ba_i^*\bz\bv^*\!\ba_i+\ba_i^*\bv\bz^*\!\ba_i\!)^2}{2|\ba_i^*\bz|^2}$, we will use a two-step procedure: First, we will show that for any fixed $\bv\in\bbC^n$ with $\|\bv\|=1$, there exists some $c_4>0$ such that
\begin{equation}\label{eq:singlev}
\Pr\left\{\!\!\frac{(\ba_i^*\bz\bv^*\!\ba_i\!\!+\!\!\ba_i^*\bv\bz^*\!\ba_i\!)^2}{2|\ba_i^*\bz|^2}\leq c_4\frac{m}{n}\!\!\right\}\geq 1-C\exp(-Cm).\end{equation} Second, we will apply an $\epsilon$-net argument to the set $\{\bv\in\bbC^n: \|\bv\|=1\}$.

To prove \eqref{eq:singlev}, considering that the expression only depends on the inner products $\ba_i^*\bv$ and $\ba_i^*\bz$, it is equivalent to work with $\tilde{\ba}_i=P_{\Sp(\bv,\bz)}\ba_i$, $\tilde{\bv}=P_{\Sp(\bv,\bz)}\bv$, and $\tilde{\bz}=P_{\Sp(\bv,\bz)}\bz$, that is, $\tilde{\ba}_i, \tilde{\bv}, \tilde{\bz}\in\bbC^2$ are obtained by projecting these vectors to the subspace spanned by $\bv$ and $\bz$. Then for any $1\leq k\leq K$, $[\tilde{\ba}_{(k-1)n+1},\cdots,\tilde{\ba}_{kn}]\in\bbC^{n\times 2}$ is a random orthogonal matrix in $\bbC^{n\times 2}$. While $\{\tilde{\ba}_i\}_{i=1}^m$ are not independently distributed, their correlation is weak and can be decomposed as follows: For $1\leq k\leq K$, let $\bS_k$ be a random matrix that represents the covariance of a Gaussian matrix of $n\times 2$, then $\{\bb_i\}_{i=1}^m$ defined by $\bb_i=\tilde{\ba}_i\sqrt{\bS_{\lfloor (i-1)/n\rfloor+1}}$ are i.i.d. sampled from $N(0,\bI_{2\times 2})$. Let $t$ be a parameter such that
\begin{equation}\label{eq:event}
\Pr\{\max_{1\leq k\leq K}\|\sqrt{\bS_{k}}/\sqrt{n}-\bI\|\leq t\}\geq 1/2.
\end{equation}
By the property of the singular values of a random Gaussian matrix \cite[Theorem 2.13]{Szarek2001} and a union bound over $K=\frac{m}{n}$ matrices $\{\bS_k\}_{k=1}^K$, we may let $t=\frac{c\log \frac{m}{n}}{\sqrt{n}}$, which converges to zero as $n,m\rightarrow\infty$ since  $\sqrt{n}>\log^2 m$. We remark that
\begin{equation}\label{eq:correlation}
\text{when the event in \eqref{eq:event} holds, $\|\tilde{\ba}_i-\frac{1}{\sqrt{n}}\bb_i\|\leq t\|\tilde{\ba}_i\|$ for all $1\leq i\leq m$.}
\end{equation}

Let $\calI_1=\{1\leq i\leq m: \|\tilde{\ba}_i^*\|\geq c_5/\sqrt{n}, |\tilde{\ba}_i^*\bz|\geq c_6\|\tilde{\ba}_i^*\|, |\tilde{\ba}_i^*\bv|\geq c_6\|\tilde{\ba}_i^*\|, \frac{\mathrm{Re}(\tilde{\ba}_i^*\tilde{\bz}\tilde{\bv}^*\!\tilde{\ba}_i)}{\|\tilde{\bv}^*\!\tilde{\ba}_i\|\|\tilde{\bz}^*\!\tilde{\ba}_i\|}\geq c_7\}$, and let $\calI_2=\{1\leq i\leq m: \|{\bb}_i^*\|\geq c_5(1-t),  |{\bb}_i^*\bz|\geq (c_6-t)\|{\bb}_i^*\|, |{\bb}_i^*\bv|\geq (c_6-t)\|{\bb}_i^*\|, \frac{\mathrm{Re}({\bb}_i^*\tilde{\bz}\tilde{\bv}^*\!{\bb}_i)}{\|\tilde{\bv}^*\!{\bb}_i\|\|\tilde{\bz}^*\!{\bb}_i\|}\geq c_7-3t/c_6\}$. By definition, when the event in \eqref{eq:event} holds, then
${\calI_1}\subseteq  \calI_2$ and $\bar{\calI_2}\supseteq \bar{\calI_1}$. Since the event in \eqref{eq:event} is independent of $\{\tilde{\ba}_i\}_{i=1}^m$, for any $\alpha>0$,
\begin{align*}
&\Pr\{|\bar{\calI}_2|\geq \alpha m\}\geq \Pr\{|\bar{\calI}_1|\geq \alpha m\}\Pr\{\max_{1\leq k\leq K}\|\bS_k-n\bI\|/n\leq t\}\\\geq& \frac{1}{2}\Pr\{|\bar{\calI}_1|\geq \alpha m\}.
\end{align*}
On the other hand, clearly we can choose $c_5, c_6, c_7, 0<\alpha<1$ such that
\[
\Pr\{|\bar{\calI}_2|\geq \alpha m\}=\Pr\{|{\calI}_2|\leq (1-\alpha)m\}\leq C\exp(-Cm),
\]
and it suggests that
\[
\Pr\{|{\calI}_1|\leq (1-\alpha)m\}=\Pr\{|\bar{\calI}_1|\geq \alpha m\}\leq 2\Pr\{|\bar{\calI}_2|\geq \alpha m\}\leq C\exp(-Cm).
\]

As a result, with probability $1-C\exp(-Cm)$, $|{\calI}_1|\geq (1-\alpha)m$ and as a result,
\[
\sum_{i=1}^m\!\frac{(\tilde{\ba}_i^*\tilde{\bz}\tilde{\bv}^*\!\tilde{\ba}_i\!\!+\!\!\tilde{\ba}_i^*\tilde{\bv}\tilde{\bz}^*\!\tilde{\ba}_i\!)^2}{2|\tilde{\ba}_i^*\bz|^2}\geq (1-\alpha)c_5^2c_6^2c_7^2\frac{m}{n},
\]
and \eqref{eq:singlev} is proved with $c_4=(1-\alpha)c_5^2c_6^2c_7^2$.

It remains to apply an $\epsilon$-net argument, which is a standard argument: by noting that
\[
\sum_{i=1}^m\Big|\frac{(\ba_i^*\bz\bv^*\!\ba_i\!\!+\!\!\ba_i^*\bv\bz^*\!\ba_i\!)^2}{2|\ba_i^*\bz|^2}-\frac{(\ba_i^*\bz\bv'^*\!\ba_i\!\!+\!\!\ba_i^*\bv'\bz^*\!\ba_i\!)^2}{2|\ba_i^*\bz|^2}\Big|\leq \sum_{i=1}^m\|\ba_i\|^2\|\bv-\bv'\|^2=m\|\bv-\bv'\|^2
\]
and use $\epsilon=\frac{c_4}{2\sqrt{n}}$, we have \eqref{eq:unitary2}.

\subsection{Proof of \eqref{eq:unitary3}}
The proof of \eqref{eq:unitary3} will be similar to the proof of \eqref{eq:unitary2}. First, it will be proved for a fixed $\bv$, and then an $\epsilon$-net argument will be used. In the proof of \eqref{eq:unitary3} for any fixed $\bv$, we apply $\tilde{\ba}_i$ and $\bb_i$ as defined in the proof of \eqref{eq:unitary2}.

For any fixed $\bv$, the first part of the proof of Lemma~\ref{lemma:term34} implies that it is sufficient to assume $\bv\perp\bz$ and $\tilde{\bv}\perp\tilde{\bz}$. When $i\in\calS(\bv,\beta)$, then $\beta\|\tilde{\ba}_i^*\tilde{\bv}\|\geq \|\tilde{\ba}_i^*\tilde{\bz}\|$. When the event \eqref{eq:event} holds, then \eqref{eq:correlation} suggests that $\frac{\|\bb_i^*\tilde{\bv}\|}{\|\bb_i^*\tilde{\bz}\|}\geq \frac{1-t}{\beta+t}$, which happens with probability smaller than $(\frac{\beta+t}{1-t})^2$ by Lemma~\ref{lemma:term34}.

Applying Hoeffding's inequality for subgaussian distributions, there exists $\gamma_0>0$ such that with probability $C\exp(-Cm)$, $\sum_{1\leq i\leq m: \frac{\|\bb_i^*\tilde{\bv}\|}{\|\bb_i^*\tilde{\bz}\|}\geq \frac{1-t}{\beta+t}}|\bb_i^*\tilde{\bv}|^2\geq (\frac{\beta+t}{1-t})^2\gamma_0 m$. By the same argument as the proof of \eqref{eq:singlev} and $|\bb_i^*\tilde{\bv}|^2-|\tilde{\ba}_i^*\tilde{\bv}|^2\leq t^2\|\tilde{\ba}_i\|^2$, we have that with probability at most $2C\exp(-Cm)$,
\[
\sum_{i\in\calS(\bv,\beta)}|\tilde{\ba}_i^*\bv|^2\geq (\frac{\beta+t}{1-t})^2\gamma_0 \frac{m}{n}+t^2\frac{m}{n}.
\]

Combining it with an $\epsilon$-net argument with $\epsilon=c/n$ and Lemma~\ref{lemma:term31} (with $c_2=n/2c$), note that
\[
\{1\leq i\leq m: \|\ba_i\|/|\ba_i^*\bv|\geq c_2\}\|\ba_i^*\bv\|^2\leq \{1\leq i\leq m: \|\ba_i\|/|\ba_i^*\bv|\geq c_2\}\|\ba_i\|^2/c_2^2\leq m/c_2^2,
\]
\eqref{eq:unitary3} is proved.

\section{Discussion}
\subsection{Comparison with existing analysis of real-valued objects}
This section compares the analysis in this case with the analysis of the same algorithm of real-valued objects in Tan and Vershynin \cite{10.1093/imaiai/iay005}, since both works have deterministic conditions of convergence and verify the deterministic condition under a probabilistic model.

First, the deterministic condition in \cite{10.1093/imaiai/iay005} can be rewritten as follows: there exist $\theta$ such that for all ``wedges of angle $\theta$'' $\calW$ in $\reals^n$,
\begin{equation}\label{eq:compare1}
\frac{1}{m}\lambda_{\min}\Big(\sum_{i=1}^m\ba_i\ba_i^T-4\sum_{\ba_i\in \calW}\ba_i\ba_i^T\Big) \geq \frac{c}{n},
\end{equation}
and here wedge of angle $\theta$ represents the region of the sphere between two hemispheres with normal vectors making an angle of $\theta$.

In comparison,  combining Theorems~\ref{thm:main1} and~\ref{thm:main2}, the deterministic result in this paper  requires the existence of some $c_0>0$ and $\alpha>1$ such that
\begin{equation}\label{eq:compare2}
  \min_{\|\bv\|=1} \!\left\{\!\frac{1}{2}\!\sum_{i=1}^m\!\frac{(\ba_i^*\bz\bv^*\!\ba_i\!+\!\ba_i^*\bv\bz^*\!\ba_i)^2}{2|\ba_i^*\bz|^2}\!-\! \frac{6}{\alpha\!-\!1}\!\!\sum_{i=1}^m|\ba_i^*\bv|^2\!-\!(2\!+\!4\alpha\!)\!\!\!\!\!\!\sum_{i\in\calS(\bv,c_0\alpha)}\!\!\!\!\!\!|\ba_i^*\bv|^2\!\right\}\!\geq \! \frac{c}{n}.
\end{equation}

The term $\sum_{i=1}^m\ba_i\ba_i^T$ in \eqref{eq:compare1} is comparable to the terms $\frac{1}{2}\sum_{i=1}^m\frac{(\ba_i^*\bz\bv^*\ba_i+\ba_i^*\bv\bz^*\ba_i)^2}{2|\ba_i^*\bz|^2}- \frac{6}{\alpha-1}\sum_{i=1}^m|\ba_i^*\bv|^2$ in \eqref{eq:compare2}. Under the real-valued setting, the latter can be simplified to $\left(1-\frac{6}{\alpha-1}\right)\sum_{i=1}^m|\ba_i^*\bv|^2$, and minimizing it over all $\|\bv\|=1$ gives the smallest eigenvalue of $\sum_{i=1}^m\ba_i\ba_i^T$.

The term $\sum_{\ba_i\in \calW}\ba_i\ba_i^T$ in \eqref{eq:compare1} and the set $\calW$ are also comparable to the term $(2+4\alpha)\sum_{i\in\calS(\bv,c_0\alpha)}|\ba_i^*\bv|^2$ in \eqref{eq:compare2} and the set $\calS(\bv,c_0\alpha)$.
In fact, the set $\calS(\bv,c_0\alpha)$ also has the ``wedge'' shape under the real-valued setting, and both works attempt to show that the number of sensing vectors in the set is small.

Second, the probabilistic analysis in these two works also shares connections, which is natural since there are similarities in the deterministic conditions. However, \cite{10.1093/imaiai/iay005} achieves the bound $m=O(n)$, which is a logarithmic factor better than the bound $m=O(n\log n)$ in Theorem~\ref{thm:main}. Looking into the analysis of both works, the extra $\log n$ factor comes from the estimation of $\sum_{i=1}^m\frac{(\ba_i^*\bz\bv^*\ba_i+\ba_i^*\bv\bz^*\ba_i)^2}{2|\ba_i^*\bz|^2}$ in Lemma~\ref{lemma:term1} and the estimation of the size of $\calS(\bv,c_0\alpha)$ in Lemma~\ref{lemma:term3}, where simple $\epsilon$-net arguments are used. In comparison, in the real-valued setting \cite{10.1093/imaiai/iay005},  $\sum_{i=1}^m\frac{(\ba_i^*\bz\bv^*\ba_i+\ba_i^*\bv\bz^*\ba_i)^2}{2|\ba_i^*\bz|^2}=2\sum_{i=1}^m\|\ba_i^*\bv\|^2$ and a standard result on the eigenvalue of $\sum_{i=1}^m\ba_i\ba_i^*$ can be used; and the number of the sensing vectors in the set $\calW$ is uniformly bounded by applying VC theory, and these arguments do not have natural generalizations to the complex-valued setting. In comparison, the $\epsilon$-net argument gives an additional $\log n$ factor. It would be interesting to investigate whether there exist more careful arguments for the complex-valued setting such that the $\log n$ factor could be removed.

Finally, while there are similarities between this work and \cite{10.1093/imaiai/iay005}, the fundamental difference comes from the argument for the deterministic condition in Theorems~\ref{thm:main1} and~\ref{thm:main2}, which relates the convergence of the randomized Kaczmarz algorithm with the local convexity of an objective function. In comparison, the straightforward calculation in \cite{10.1093/imaiai/iay005} is based on the fact that there only exist  two phases of $\pm 1$ in the real-valued setting.

\subsection{Initialization}
Theorem~\ref{thm:main} requires an initialization such that $\|\bx^{(0)}-\bz\|\leq c_0\sqrt{\delta_1}$. Many schemes have been proposed for obtaining a good initialization \cite[Section B]{10.1093/imaiai/iay005}. For example, we may use the truncated spectral method that let $\hat{\bx}^{(0)}=\lambda_0\tilde{\bx}$, where $\lambda_0=\sqrt{\frac{1}{m}\sum_{i=1}^m b_i^2}$ and $\tilde{\bx}$ in the leading eigenvector of $
Y=\frac{1}{m}\sum_{i=1}^mb_i^2\ba_i\ba_i^* I(b_i\leq 3\lambda_0).$ Following the analysis in \cite[Section B]{10.1093/imaiai/iay005}, one can show that the requirement on the initialization $\|\hat{\bx}^{(0)}-\bz\|\leq c_0\sqrt{\delta_1} \|\bz\|$ holds as long as $m\geq C(\log(1/\delta)+n)/( c_0^2\delta_1)$.

However, considering that this construction of the initialization is  dependent on the isotropy of the distribution of the sensing vectors, it is still interesting to investigate whether the randomized Kaczmarz algorithm works with random initialization, similar to the results in \cite{tan2019online,zhang2020}, and we leave it as a possible future direction.

\subsection{Other probabilistic models}
Sections~\ref{sec:step3} and~\ref{sec:step4} verifies the deterministic condition in Theorems~\ref{thm:main1} and~\ref{thm:main2} when the sensing vectors are sampled i.i.d. from a uniform distribution on the sphere or generated from a unitary model. However, there exist many other models of generating sensing vectors, such as  coded diffraction model in \cite{Wei_2015}, and it would be interesting to see whether the deterministic condition in Theorems~\ref{thm:main1} and~\ref{thm:main2} could be verified for more generic models. 

\section{Summary}
This paper justifies the convergence of the randomized Kaczmarz algorithm for phase retrieval of complex-valued objects. Specifically, the paper first establishes a deterministic condition for its convergence, and then demonstrates that  when the sensing vectors are sampled uniformly from a unit sphere in $\bbC^n$ and the number of sensing vectors $m$ satisfies $m>O(n\log n)$ as $n,m\rightarrow\infty$, then this deterministic condition holds with high probability. 


%
%

\bibliographystyle{spmpsci}      
\bibliography{bib-online}   


\end{document}